\documentclass[11pt,a4paper]{llncs}

\usepackage{fullpage}
\usepackage{graphicx}
\usepackage{amsmath}
\usepackage{amssymb}
\usepackage{algorithm}
\usepackage{algorithmic}
\usepackage{calc}
\usepackage{mathtools}
\usepackage{ragged2e}
\usepackage{enumitem}

\usepackage{newfloat}

\DeclareFloatingEnvironment[fileext=frm,placement={!ht},name=Frame]{algfloat}

\newcommand{\mypara}[1]{\medskip\noindent\textbf{#1.}}  % Paragraph headers

\newcommand{\Order}{\mathrm{O}}
\newcommand{\OrderT}{\tilde{\mathrm{O}}}
\newcommand{\ThetaT}{\tilde{\mathrm{\Theta}}}
\newcommand{\polylog}{\mathop{\mathrm{polylog}}\nolimits}
\newcommand{\poly}{\mathop{\mathrm{poly}}\nolimits}

\newcommand{\CONGESTED}{\textsf{CONGESTED-CLIQUE} }
\newcommand{\CONGEST}{\textsf{CONGEST} }
\newcommand{\LOCAL}{\textsf{LOCAL} }

\title{MIS in the Congested Clique Model in $\Order(\log \log \Delta)$ Rounds\thanks{C.~Konrad is supported by the Centre for Discrete Mathematics and its Applications (DIMAP) at Warwick University and by EPSRC award EP/N011163/1.}}

\author{Christian Konrad}
\institute{Department of Computer Science, Centre for Discrete Mathematics and its Applications (DIMAP), University of Warwick, Coventry, UK, \email{c.konrad@warwick.ac.uk}}

\newcommand{\Exp}{\mathbb{E}}

\renewcommand{\Pr}{\mathbb{P}}

\begin{document}

 \maketitle

 \begin{abstract}
  %The congested clique model, network nodes 
  We give a maximal independent set (MIS) algorithm that runs in $\Order(\log \log \Delta)$ rounds in the congested clique model, where $\Delta$ is the maximum degree of the input graph. This improves upon the $\Order(\frac{\log(\Delta) \cdot \log \log \Delta}{\sqrt{\log n}} + \log \log \Delta )$ rounds algorithm of [Ghaffari, PODC '17], where $n$ is the number of vertices of the input graph.
  
  In the first stage of our algorithm, we simulate the first $\Order(\frac{n}{\poly \log n})$ iterations of the sequential random order Greedy algorithm for MIS  in the congested clique model in $\Order(\log \log \Delta)$ rounds. This thins out the input graph relatively quickly: After this stage, the maximum degree of the residual graph is poly-logarithmic. In the second stage, we run the MIS algorithm of [Ghaffari, PODC '17] on the residual graph, which completes in $\Order(\log \log \Delta)$ rounds on graphs of poly-logarithmic degree.  
 \end{abstract}

%  \begin{center}
%  To be considered as a regular paper and as a brief announcement. 
  
%  \vspace{0.3cm} 
  
%  Not eligible for best student paper award.
% \end{center}

% \newpage
 
% \clearpage
%\setcounter{page}{1}

\section{Introduction}
\mypara{The {\normalfont \LOCAL}and {\normalfont \CONGEST}Models} The \LOCAL \cite{l87,p00} and \CONGEST \cite{p00} models 
are the most 
studied computational models for distributed graph algorithms. In these models, a communication network is represented by an $n$-vertex 
graph $G=(V, E)$, which also constitutes the input to a computational graph problem. Each vertex (or network node) $v \in V$ hosts a 
computational unit and is identified by a unique ID $\in \Theta(\log n)$. Initially, besides its ID, every vertex knows its neighbors 
(and their IDs). All network nodes simultaneously commence the execution of a distributed algorithm. Such an algorithm proceeds in 
synchronous rounds, where each round consists of two phases. In the computation phase, every vertex may execute unlimited computations. 
This is followed by the communication phase, where vertices may exchange individual messages with their neighbors. While message lengths 
are unbounded in the \LOCAL model, in the \CONGEST model every message is of length $\Order(\log n)$. The goal is to 
design algorithms that employ as few communication rounds as possible. The output is typically distributed. For independent set problems, 
which are the focus of this paper, upon termination of the algorithm, every vertex knows whether it participates in the independent set.

The \LOCAL model provides an abstraction that allows for the study of the {\em locality} of a distributed problem, i.e., how 
far network nodes need to be able to look into the network in order to complete a certain task. In addition to the locality constraint, the 
\CONGEST model also 
addresses the issue of {\em congestion}. For example, while in the \LOCAL model, network nodes can learn their distance-$r$ 
neighborhoods in $r$ rounds, this is generally not possible in the \CONGEST model due to the limitation of message sizes.

\mypara{The {\normalfont \CONGESTED}Model} In recent years, the \CONGESTED model \cite{lppp05}, a variant of the \CONGEST model,
has received significant attention (e.g. \cite{jn18,g17,ks17,gp16,l16,ckklps15,hps14,dko14,hp14,l13}). It differs from the \CONGEST model in that every pair of vertices (as opposed 
to only every pair of adjacent vertices) can exchange messages of sizes $\Order(\log n)$ in the communication phase. The focus of this 
model thus solely lies on the issue of congestion, since non-local message exchanges are now possible. This model 
is at least as powerful as the \CONGEST model, and many problems, such as computing a minimum spanning tree \cite{jn18,gp16} or computing the size of
a maximum matching \cite{l16}, can in fact be 
solved much faster than in the \CONGEST model. In \cite{g17}, Ghaffari asks whether any of the classic local problems - maximal independent 
set (MIS), maximal matching, ($\Delta+1$)-vertex-coloring, and ($2\Delta-1$)-edge-coloring - can be solved much faster in the 
\CONGESTED model than in the \CONGEST model, 
where $\Delta$ is the maximum degree of the input graph. Ghaffari made progress on this question and gave a 
$\Order(\frac{\log(\Delta) \cdot \log \log \Delta}{\sqrt{\log n}} + \log \log \Delta )$ rounds MIS algorithm in the {\sf CONGESTED-CLIQUE}
model, while 
the best known \CONGEST model algorithm runs in $\Order(\log \Delta) + 2^{\Order(\sqrt{\log \log n})}$ rounds \cite{g16}. This algorithm separates the two models 
with regards to the MIS problem, since it is known that 
$\Omega(\min \{ \frac{\log \Delta}{\log \log \Delta}, \sqrt{\frac{\log n}{\log \log n}} \})$ rounds are required for MIS in the 
\CONGEST model \cite{kmw04,kmw16} \footnote{This lower bound even holds in the \LOCAL model.}.

\mypara{Result}
While Ghaffari gave a roughly quadratic improvement over the best \CONGEST model MIS algorithm, in this paper, we show that an exponential 
improvement is possible. Our main result is as follows:

\begin{theorem}[Main Result] \label{thm:main}
 Let $G=(V, E)$ be a graph with maximum degree $\Delta$. There is a randomized algorithm in the {\normalfont \CONGESTED}model 
 that operates in (deterministic) $\Order(\log \log \Delta)$ rounds and outputs a maximal independent set in $G$ with high probability.
\end{theorem}

%\noindent Based on well-known reductions \cite{l92}, our algorithm can also be used to solve the maximal
%matching, ($\Delta+1$)-vertex-coloring, and ($2\Delta-1$)-edge-coloring problems in the same number of rounds.

\mypara{Techniques} Ghaffari gave a variant of his MIS algorithm that runs in $\Order(\log \log \Delta)$ rounds
on graphs $G$ with poly-logarithmic maximum degree, i.e., $\Delta(G) = \Order(\polylog n)$ (Lemma 2.15. in \cite{g17}) \footnote{This variant works in fact on graphs with 
maximum degree bounded by $2^{c \sqrt{\log n}}$, for a sufficiently small constant $c$, but a poly-logarithmic degree bound is sufficient for our purposes.}. To achieve
a runtime of $\Order(\log \log \Delta)$ rounds even on graphs with arbitrarily large maximum degree, we give a $\Order(\log \log \Delta)$ rounds 
algorithm that computes an independent set $I$ such that the residual graph $G \setminus \Gamma_G[I]$ ($\Gamma_G[I]$ denotes the inclusive neighborhood of $I$ in $G$) 
has poly-logarithmic maximum degree. 
We then run Ghaffari's algorithm on the residual graph to complete the independent set computation. 

Our algorithm is an implementation of the sequential \textsc{Greedy} algorithm for MIS in the {\sf{CONGESTED-CLIQUE}} model. 
\textsc{Greedy} processes the vertices of the input graph in arbitrary order and adds the current vertex to an initially empty independent set
if non of its neighbors have previously been added. The key idea is to simulate multiple iterations 
of \textsc{Greedy} in $\Order(1)$ rounds in the {\sf{CONGESTED-CLIQUE}} model. A simulation of $\sqrt{n}$ iterations 
in $\Order(1)$ rounds can be done as follows: Let $v_1 v_2 \dots v_n$ be an arbitrary ordering of the vertices (e.g. by their IDs). 
Observe that the subgraph $G[\{v_1, \dots, v_{\sqrt{n}} \}]$ induced by the first $\sqrt{n}$ vertices has at most $n$ edges. Lenzen 
gave a routing protocol that can be used to collect these $n$ edges at one distinguished vertex $u$ in $\Order(1)$ rounds. Vertex $u$ then simulates 
the first $\sqrt{n}$ iterations of \textsc{Greedy} locally (observe that the knowledge of $G[\{v_1, \dots, v_{\sqrt{n}} \}]$ is sufficient to do this) 
and then notifies the nodes chosen into the independent set about their selection.

The presented simulation can be used to obtain a $\Order(\sqrt{n})$ rounds MIS algorithm in the {\sf{CONGESTED-CLIQUE}} model. To reduce 
the number of rounds to $\Order(\log \log n)$, we identify a {\em residual sparsity} property of the
\textsc{Greedy} algorithm: If \textsc{Greedy} processes the vertices in {\em uniform random order}, then the maximum degree of the residual graph after 
having processed the $k$th vertex is $\Order(\frac{n}{k} \log n)$ with high probability (\textbf{Lemma~\ref{lem:thinning}}). 
To make use of this property, we will thus first compute a uniform random ordering of the vertices. Then, after having processed the first $\sqrt{n}$ vertices 
as above, the maximum degree in the residual graph is $\OrderT(\sqrt{n})$\footnote{We use the notation $\OrderT(.)$, which equals the usual $\Order()$ notation 
where all poly-logarithmic factors are ignored.}. This allows us to increase the block size and simulate the next $\OrderT(n^{3/4})$ iterations in $\Order(1)$ rounds: 
Using the fact that the maximum degree in the residual graph is $\OrderT(\sqrt{n})$, it is not hard to see that the subgraph induced by the next 
$\ThetaT(n^{\frac{3}{4}})$ random vertices has a maximum degree of $\OrderT(n^{1/4})$ with high probability (and thus contains $\Order(n)$ edges). Pursuing this approach further, we can process $\ThetaT(n^{1 - \frac{1}{2^i}})$ vertices in the $i$th block, since, 
by the residual sparsity lemma, the maximum degree in the $i$th residual graph is $\OrderT(n^{\frac{1}{2^i}})$. Hence, after having processed $\Order(\log \log n)$ blocks, the 
maximum degree becomes poly-logarithmic. In Section~\ref{sec:alg}, we give slightly more involved arguments that show 
that $\Order(\log \log \Delta)$ iterations (as opposed to $\Order(\log \log n)$ iterations) are in fact enough.

\mypara{The Residual Sparsity Property of {\normalfont \textsc{Greedy}}} The author is not aware of any work that exploits or mentions 
the residual sparsity property of the random order \textsc{Greedy} algorithm for MIS. In the context of correlation clustering in the 
data streaming model, a similar property of a Greedy clustering algorithm was used in \cite{acgmw15} (Lemma 19). Their lemma is in fact 
strong enough and can give the version required in this paper. Since \cite{acgmw15} does not provide a proof, and the residual sparsity 
property is central to the functioning of our algorithm, we give a proof that follows the main idea of \cite{acgmw15} 
\footnote{The authors of \cite{acgmw15} kindly shared an extended version of their paper with me.} adapted to our needs.

\mypara{Further Related Work} The maximal independent set problem is one of the classic symmetry breaking problems in distributed computing. 
Without all-to-all communication, Luby \cite{l85} and independently Alon et al. \cite{abi86} gave $\Order(\log n)$ rounds distributed algorithms more than 30 years ago. 
Barenboim et al. \cite{beps16} improved on this for certain ranges of $\Delta$ and gave a $\Order(\log^2 \Delta)+ 2^{\Order(\sqrt{\log \log n})}$ rounds algorithm. 
The currently fastest algorithm is by Ghaffari \cite{g16} and runs in $\Order(\log \Delta)+ 2^{\Order(\sqrt{\log \log n})}$ rounds. 

The only MIS algorithm designed in the {\sf{CONGESTED-CLIQUE}} model is the previously mentioned algorithm by Ghaffari \cite{g17}. 
Ghaffari shows how multiple rounds of a {{\sf CONGEST}} model algorithm can be simulated in much fewer rounds in the 
{\sf{CONGESTED-CLIQUE}} model. This is similar to the approach taken in this paper, however, while in our algorithm the simulation of multiple iterations of the 
sequential \textsc{Greedy} algorithm is performed at one distinguished node, every node participates in the simulation
of the {{\sf CONGEST}} model algorithm in Ghaffari's algorithm.

\mypara{Outline}
We proceed as follows. First, we give necessary definitions and notation, and we state known results that we employ in this 
paper (Section~\ref{sec:prelim}). We then give a proof of the residual sparsity property of the sequential \textsc{Greedy} algorithm (Section~\ref{sec:thinning}).
Our $\Order(\log \log \Delta)$ rounds MIS algorithm is subsequently presented (Section~\ref{sec:alg}), followed by a brief conclusion (Section~\ref{sec:conclusion}).

\section{Preliminaries} \label{sec:prelim}
We assume that $G=(V, E)$ is a simple unweighted $n$-vertex graph. For a node $v \in V$, we write $\Gamma_G(v)$ to denote
$v$'s (exclusive) neighborhood, and we write $\deg_G(v) := |\Gamma_G(v)|$. The inclusive neighborhood is defined as $\Gamma_G[v] := \Gamma(v) \cup \{v\}$.
Inclusive neighborhoods are extended to subsets $U \subseteq V$ as $\Gamma_G[U] := \cup_{u \in U} \Gamma_G[u]$. Given a subset of vertices 
$U \subseteq V$, the subgraph induced by $U$ is denoted by $G[U]$.

\mypara{Independent Sets} An {\em independent set} $I \subseteq V$ is a subset of non-adjacent
vertices. An independent set $I$ is {\em maximal} if for every $v \in V \setminus I$, $I \cup \{ v\}$ is not an independent set.
Given an independent set $I$, we call the graph $G' = G[V \setminus \Gamma_{G}[I] ]$ the residual graph with respect to $I$.
If clear from the context, we may simple call $G'$ the residual graph. We say that a vertex $u \in V$ is {\em uncovered} with 
respect to $I$, if $u$ is not adjacent to a vertex in $I$, i.e., $u \in V \setminus \Gamma_G[I]$. Again, if clear from the context,
we simply say $u$ is uncovered without specifying $I$ explicitly.

Ghaffari gave the following result that we will reuse in this paper:

\begin{theorem}[Ghaffari \cite{g17}]
 Let $G$ be a $n$-vertex graph with $\Delta(G) = \poly \log (n)$. Then there is a distributed algorithm that runs in the 
 {\em \CONGESTED}model and computes a MIS on $G$ in $\Order(\log \log \Delta)$ rounds.
\end{theorem}

\mypara{Routing} As a subroutine, our algorithm needs to solve the following simple routing task: Let $u \in V$ be an arbitrary vertex. 
Suppose that every other vertex $v \in V \setminus \{u \}$ holds $0 \le n_{v} \le n$ messages each of size $\Order(\log n)$ that it wants to 
deliver to $u$. We are guaranteed that $\sum_{v \in V} n_v \le n$. Lenzen proved that in the \CONGESTED model there is a deterministic 
routing scheme that achieves this task in $\Order(1)$ rounds \cite{l13}. In the following, we will refer to this scheme as Lenzen's 
routing scheme.

\mypara{Concentration Bound for Dependent Variables} In the analysis of our algorithm, we require a Chernoff bound for dependent
variables (see for example \cite{fkv11}):

\begin{theorem}[Chernoff Bound for Dependent Variables, e.g. \cite{fkv11}] \label{thm:chernoff}
 Let $X_1, X_2, \dots, X_n$ be $0/1$ random variables for which there is a $p \in [0, 1]$ such that for all $k \in [n]$
 and all $a_1, \dots, a_{k-1} \in \{0, 1\}$ the inequality 
 $$\Pr \left[ X_k = 1 \, | \, X_1 = a_1, X_2 = a_2, \dots, X_{k-1} = a_{k-1} \right] \le p  $$
\end{theorem}
holds. Let further $\mu \ge p \cdot n$. Then, for every $\delta > 0$:
\begin{eqnarray*}
  \Pr \left[ \sum_{i=1}^n X_i \ge (1+\delta) \mu \right] \le \left( \frac{e^{\delta}}{(1+\delta)^{1+\delta}} \right)^{\mu} \ .
\end{eqnarray*}

Last, we say that an event occurs with high probability if the probability of the event not occuring is at most $\frac{1}{n}$.

\section{Sequential Random Order Greedy Algorithm for MIS} \label{sec:thinning}
The \textsc{Greedy} algorithm for maximal independent set processes the vertices of the input graph in arbitrary
order. It adds the current vertex under consideration to an initially empty independent set $I$ if none 
of its neighbors are already in $I$. %This process can also be understood as follows: Select an arbitrary
%vertex and put it into the independent set. Then, remove this vertex and its neighborhood from the input
%graph and repeat. 

This algorithm progressively thins out the input graph, and the rate at which the graph loses edges depends 
heavily on the order in which the vertices are considered. If the vertices are processed in uniform random order 
(Algorithm~1), then the number of edges in the residual graph decreases relatively quickly. A variant of the next 
lemma was proved in \cite{acgmw15} in the context of correlation clustering in the streaming model:%, however, since
%we are not aware of a lemma that directly targets the MIS problem, we give a self-contained proof here.

\begin{algfloat} 
 \normalsize
 \noindent \fbox{\parbox{\textwidth-5pt}{
 
 \textbf{Input:} $G=(V, E)$ is an $n$-vertex graph

 \begin{enumerate}
  \item Let $v_1, v_2, \dots, v_n$ be a uniform random ordering of $V$
  \item $I \gets \{ \}$, $U \gets V$ \hfill ($U$ is the set of uncovered elements)
  \item \textbf{for} $i \gets 1, 2, \dots, n$ \textbf{do} 
  
  $\quad$ \textbf{if} $v_i \in U$ \textbf{then} 
  
  $\quad\quad$ $I \gets I \cup \{v_i \}$ 
  
  $\quad\quad$ $U \gets U \setminus \Gamma_G[v_i]$  
  \item \textbf{return} $I$
 \end{enumerate} 
}}
\begin{center}
 \textbf{Algorithm 1.} Random order \textsc{Greedy} algorithm for MIS.
\end{center}
\end{algfloat}

\begin{lemma}\label{lem:thinning}
 Let $t$ be an integer with $1 \le t < n$. Let $U_i$ be the set $U$ at the beginning of iteration $i$ 
 of Algorithm~1. Then with probability at least $1 - n^{-9}$ the following holds: 
 $$\Delta(G[U_{t}]) \le 10 \ln(n) \frac{n}{t} \ .$$
\end{lemma}
\begin{proof}
 Fix an arbitrary index $j \ge t$. We will prove that either vertex $v_j$ is not in $U_t$, or it has at most 
 $10 \ln(n) \frac{n}{t}$ neighbors in $G[U_{t}]$, with probability at least $1-n^{-10}$. The result follows by a union bound
 over the error probabilities of all $n$ vertices. %The result then follows, since there 
 %are at most $n$ vertices in $G[U_t]$, and by the fact that the number of edges in a graph is half the degree 
 %sum of the vertices.
 
 We consider the following process in which the random order of the vertices is determined. First, reveal $v_j$. Then,
 reveal vertices $v_i$ just before iteration $i$ of the algorithm. Let $N_i := \Gamma_G(v_j) \cap U_i$ be the set of neighbors
 of $v_j$ that are uncovered in the beginning of iteration $i$, and let $d_i = |N_i|$. For every 
 $1 \le i \le t-1$, the following holds:
 \begin{eqnarray*}
  \Pr \left[ v_i \in N_i \, | \, v_j, v_1, \dots, v_{i-1} \right] = \frac{d_i}{n-1-(i-1)}  \ge \frac{d_i}{n} \ ,
 \end{eqnarray*}
 since $v_i$ can be one of the not yet revealed $n-1-(i-1)$ vertices.
We now distinguish two cases. First, suppose that $d_{t-1} \le 10 \ln(n) \frac{n}{t}$. Then the result follows immediately
since, by construction, $d_t \le d_{t-1}$ (the sequence $(d_i)_i$ is decreasing). Suppose next 
that $d_{t-1} > 10 \ln(n) \frac{n}{t}$. Then, we will prove that with high probability
there is one iteration $i' \le t-1$ in which a neighbor of $v_j$ is considered by the algorithm, i.e., $v_{i'} \in N_{i'}$. 
This in turn implies that $v_j$ is not in $U_{t}$. % and hence edges incident to $v_j$ are not in $E(G[U_{t}])$. 
We have:
\begin{eqnarray*}
\Pr \left[ \forall i < t: v_i \notin N_i \,  | \, v_j \right] & \le & \prod_{i < t} \Pr \left[ v_i \notin N_i \, | \, v_j, v_1, \dots, 
v_{i-1} \right] \le \prod_{i < t} (1 - \frac{d_i}{n}) \\
& \le & (1 - \frac{d_{t-1}}{n})^{t-1} \le e^{\frac{d_{t-1}(t-1)}{n}} \le n^{-10} \ .
\end{eqnarray*}
\qed
\end{proof}

\section{MIS Algorithm in the Congest Clique Model} \label{sec:alg}

\subsection{Algorithm}

\begin{algfloat}

 \normalsize
 \noindent \fbox{\parbox{\textwidth-5pt}{
 
 \textbf{Input:} $G=(V, E)$ is an $n$-vertex graph with maximum degree $\Delta := \Delta(G)$

 Set parameter $C = 5$
 
 \begin{enumerate}
  \item \textbf{Nodes agree on random order. } 
  
  All vertices exchange their IDs in one round. Let $u \in V$ be the vertex with the smallest ID. Vertex $u$ choses a uniform random 
  order of $V$ and informs every vertex $v \in V \setminus \{u \}$ about its position $r_v$ within the order. Then, every vertex 
  $v \in V$ broadcasts $r_v$ to all other vertices. As a result, all vertices know the order. Let $v_1, v_2, \dots, v_n$ be the resulting order.
 
 \vspace{0.3cm}
 
 \item \textbf{Simulate sequential Greedy. }
 
 Every vertex $v_i$ sets $u_i \gets true$ indicating that $v_i$ is uncovered. Let $G' := G$.
 %Every vertex $v_i$ broadcasts $u_i$ to all other vertices. Then every vertex $v_i$ computes $\deg_{G'}(v_i)$ locally. 
 
 Every vertex $v_i$ broadcasts $\deg_{G'}(v_i)$ to all other vertices so that every vertex knows $\Delta(G')$.
 
 \textbf{while} $\Delta(G') > \log^4 n$ \textbf{do} 
 
 %Let $n_0 \gets 0$ and $n_i = \lceil \frac{n^{1-\frac{1}{2^i}}}{\log n} \rceil$ for every $i \ge 1$
 
 %\textbf{For} $j = 1, 2, \dots, \lceil \log \log n \rceil$ \textbf{do}: 
 
 \begin{enumerate}
  \item Let $k \gets \frac{n}{\sqrt{\Delta(G')} C}$
  \item Every vertex $v_i$ with $u_i = true$ and $i \le k$ sends all its incident edges $v_i v_j$ with $u_j = true$
 and $j < i$ to $v_1$ using Lenzen's routing protocol in $\Order(1)$ rounds.
  \item Vertex $v_1$ knows the subgraph $H$ of uncovered vertices $v_j$ with $j \le k$, i.e., 
  $$H := G'[ \{ v_j \ : \ j \le k \mbox{ and } u_j = true \}] \ . $$ 
  It continues the simulation of \textsc{Greedy} up to iteration $k$ using $H$. Let $I'$ be the vertices selected into the independent set.
  \item Vertex $v_1$ informs nodes $I'$ about their selection in one round. Nodes $I'$ inform their neighbors about their selection in one round.
  \item Every node $v_i \in \Gamma_G[I']$ sets $u_i \gets false$.
  \item Let $G' := G[ \{ v_i \in V \, : \, u_i = true \}]$. Every vertex $v_i$ broadcasts $u_i$ to all other vertices. Then every vertex $v_i$ computes $\deg_{G'}(v_i)$ locally
  and broadcasts $\deg_{G'}(v_i)$ to all other vertices. As a result, every vertex knows $\Delta(G')$.
 \end{enumerate}

 \textbf{end while}
 
 \vspace{0.3cm}
 
 \item \textbf{Run Ghaffari's algorithm.} 
 
 Run Ghaffari's MIS algorithm on $G'$ in $\Order(\log \log \Delta)$ rounds. %Let $I_2$ be the computed independent set. 
 
 \end{enumerate} 
}}
\begin{center}
 \textbf{Algorithm 2.} $\Order(\log \log \Delta)$ rounds MIS algorithm in the \CONGESTED model.
\end{center}
\end{algfloat}

Our MIS algorithm, depicted in Algorithm~2, consists of three parts: 

First, all vertices agree on a uniform random order as follows. The vertex with the smallest ID choses a uniform random order 
locally and informs all other vertices about their positions within the order. Then, all vertices broadcast their positions to all
other vertices. As a result, all vertices know the entire order. Let $v_1, v_2, \dots, v_n$ be this order.

Next, we simulate \textsc{Greedy} until the maximum degree of the residual graph is at most $\log^4 n$ (this bound is chosen only for convenience;
any poly-logarithmic number in $n$ is equally suitable). To this end, in each iteration of the while-loop, we first determine
a number $k$ as a function of the maximum degree $\Delta(G')$ of the current residual graph $G'$ so that the subgraph of $G'$ induced 
by the yet uncovered vertices of $\{v_1, \dots, v_k\}$ has at most $n$ edges w.h.p. (see Lemma~\ref{lem:edge-bound}). Using Lenzen's routing protocol, 
these edges are collected at vertex $v_1$, which continues the simulation of \textsc{Greedy} up to iteration $k$. It then informs
the chosen vertices about their selection, who in turn inform their neighbors about their selection. Vertices then compute the new 
residual graph and its maximum degree and proceed with the next iteration of the while-loop. We prove in Lemma~\ref{lem:max-deg} that only
$\Order(\log \log \Delta)$ iterations of the while-loop are necessary until $\Delta(G')$ drops below $\log^4 n$.

Last, we run Ghaffari's algorithm on $G'$ which completes the maximal independent set computation.

\subsection{Analysis}
%Let $i_{max}$ be the number of iterations of the while-loop. For $1 \le i \le i_{max}$, let $I'_i$ be the independent set computed in iteration 
%$i$ of the while-loop. Let $I'$ be the independent set computed by Ghaffari's algorithm. Then the output independent set is $I := I' \cup \bigcup_{i=1}^{i_{max}} I'_i$.

Let $G'_i$ denote the graph $G'$ at the beginning of iteration $i$ of the while-loop. Notice that $G'_1 = G$. Let $\Delta_i := \Delta(G'_i)$
and let $k_i = \frac{n}{\sqrt{\Delta_i}C}$ be the value of $k$ in iteration $i$. Observe that the while-loop is only executed if $\Delta_i > \log^4 n$
and hence
\begin{eqnarray}
 k_i \ge \frac{n}{\log^2 n C} \label{eqn:k-bound}
\end{eqnarray}
holds for every iteration $i$ of the while-loop. Further let $H_i$ be the graph $H$ in iteration $i$ of the while-loop. 

To establish the runtime of our algorithm, we need to bound the number of iterations of the while-loop. To this end, in the next lemma we 
bound $\Delta_i$ for every $1 \le 1 \le n$ and conclude that $\Delta_{j} \le \log^4 n$, for some $j \in \Order(\log \log \Delta)$.

%We will first bound $\Delta_i$, for every $i$. This implicitly assumes that 

%In order to bound the number of iterations of the while loop, we will bound $\Delta_i$ for every $i$. 

\begin{lemma} \label{lem:max-deg}
 With probability at least $1 - n^{-8}$, for every $i \le n$, the maximum degree in $G'_i$ is bounded as follows:
 $$\Delta_i \le \Delta^{\frac{1}{2^{i-1}}} \cdot 100 C \ln^2 n \ . $$
\end{lemma}
\begin{proof}
 We prove the statement by induction. Observe that $\Delta_1 = \Delta$ and the statement is thus trivially true for $i=1$. Suppose that the statement holds up to some index $i-1$. 
 Recall that $G_i'$ is the residual graph obtained by running \textsc{Greedy} on vertices $v_1, \dots, v_{k_{i-1}}$. Hence, by applying Lemma~\ref{lem:thinning}, the following
 holds with probability $1 - n^{-9}$:
 \begin{eqnarray*}
 \Delta_i \le 10 \ln(n) \frac{n}{k_{i-1}} = \frac{10 \ln(n) n}{\frac{n}{\sqrt{\Delta_{i-1} }C}} = \sqrt{\Delta_{i-1}} \cdot 10 C \ln n  \ .
 \end{eqnarray*}
 Resolving the recursion, we obtain
 \begin{eqnarray*}
  \Delta_i  = \Delta^{\frac{1}{2^{i-1}}} \cdot \prod_{j=0}^{i-2} \left( 10 C \ln n \right)^{\frac{1}{2^j}} = 
  \Delta^{\frac{1}{2^{i-1}}} \cdot \left( 10 C \ln n \right)^{\sum_{j=0}^{i-2} \frac{1}{2^j}}  
  \le \Delta^{\frac{1}{2^{i-1}}} \cdot 100 C^2 \ln^2 n \ .
 \end{eqnarray*}
 Observe that we invoked $n$ times Lemma~\ref{lem:thinning}. Thus, by the union bound, the result holds 
 with probability $1-n^{-8}$.
\qed
\end{proof}

\begin{corollary}\label{cor:runtime}
 $\Delta_i = \Order(\log^2 n)$ for some $i \in \Order(\log \log \Delta)$.
\end{corollary}

To establish correctness of the algorithm, we need to ensure that we can apply Lenzen's routing protocol to collect the edges of $H_i$ 
at vertex $v_1$. For this to be feasible, we need to prove that, for every $i$, $H_i$ contains at most $n$ edges with high probability.

\begin{lemma} \label{lem:edge-bound}
 With probability at least $1-n^{-9}$, graph $H_i$ has at most $n$ edges.
\end{lemma}
\begin{proof}
 Let $U_i$ be the vertex set of $G_i'$, i.e., the set of uncovered vertices at the beginning of iteration $i$.  
 We will prove now that, with probability at least $1-n^{10}$, for every $v_j \in U_i$, the following holds
 \begin{eqnarray}
 \label{eqn:392}
d(v_j) := |\Gamma_{G_i'}(v_j) \cap \{v_{k_{i-1} + 1}, \dots, v_k \}| \le \frac{n}{k_i} \ .
 \end{eqnarray}
 Since the vertex set of $H_i$ is a subset of at most $k_i - k_{i-1} \le k_i$ vertices of $U_i$, the result follows
 by applying the union bound on the error probabilities for every vertex of $G_i'$.
 
 %Since the vertex set of $H_i$ is a subset of $\{v_{k_{i-1} + 1}, \dots, v_k \}$ and Inequality~\ref{eqn:392}
 %holds for the choice of $v_j$ is arbitrary (in particular, 
 %the bound also holds for the vertices of $H_i$), this implies that $\Delta(H_i) \le \frac{n}{k_i}$. The result then follows since $H_i$ has less than $k_i$ vertices.
  
 %Observe further that $k_i \le \frac{n}{\log^2 n \sqrt{2 C}}$ always holds. 
 
 To prove Inequality~\ref{eqn:392}, observe that graph $G'_i$ is solely determined by vertices 
 $v_1, v_2, \dots, v_{k_{i-1}}$, and the execution of the algorithm so far was not affected by the outcome of 
 the random variables $v_{k_{i-1}+1}, \dots, v_n$.  Thus, by the principle of deferred decision, for every $k_{i-1}+1 \le l \le k_i$,
 vertex $v_l$ can be seen as a uniform random vertex chosen from $V \setminus \{v_1, \dots, v_{l-1} \}$.
  
 For $1 \le l \le k_i - k_{i-1}$, let $X_l$ be the indicator variable of the event ``$v_{k_{i-1} + l} \in \Gamma_{G_i'}(v_j)$''.
 Observe that $d(v_j) = \sum_{l} X_l$ and 
 \begin{eqnarray}
\Exp \left[ d(v_j) \right] = \deg_{G_i'}(v_j) \cdot \frac{k_i - k_{i-1}}{n - k_{i-1}} \le \deg_{G_i'}(v_j) \cdot \frac{k_i}{n} \ .  \label{eqn:811}
 \end{eqnarray}

 Furthermore, observe that for every $1 \le l \le k_i - k_{i-1}$, and all $a_1, \dots, a_{l-1} \in \{0, 1\}$, the inequality 
 \begin{eqnarray*}
  \Pr \left[ X_l = 1 \ | \ X_1 = a_1, X_2 = a_2, \dots, X_{l-1} = a_{l-1} \right] \le \frac{\deg_{G_i'}(v_j)}{n - k_i} \le \frac{2 \cdot \deg_{G_i'}(v_j)}{n} 
 \end{eqnarray*}
holds (using the bound $k_i \le n/2$, which follows from Inequality~\ref{eqn:k-bound}), since in the worst case, we have 
$a_1 = a_2 = \dots = a_{l-1} = 0$, which implies that there are $\deg_{G_i'}(v_j)$ choices left out of at least $n-k_i$ possibilities
such that $X_l = 1$. We can thus use the Chernoff bound for dependent variables as 
stated in Theorem~\ref{thm:chernoff} in order to bound the probability that $d(v_j)$ deviates from its expectation.

We distinguish two cases. First, suppose that $\Exp \left[ d(v_j) \right] \ge 4 \log n$. Then by Theorem~\ref{thm:chernoff} (setting $\mu = 2 \Exp \left[ d(v_j) \right]$
and $\delta = 8$),
\begin{eqnarray*}
 \Pr \left[ d(v_j) \ge 18 \cdot \Exp [ d(v_j) ] \right] \le \exp \left( \frac{e^{8}}{(1+8)^{1+8}} \right)^{8 \log n} \le n^{-10} \ .
\end{eqnarray*} 
Thus, using Inequality~\ref{eqn:811}, with high probability, 
\begin{eqnarray*}
d(v_j) \le 18 \cdot \Exp \left[ d(v_j) \right] \le 18 \cdot \deg_{G_i'}(v_j) \frac{k_i}{n} \le 18 \cdot \frac{\Delta_i}{\sqrt{\Delta_i} C} \le 18 \cdot \frac{n}{k_i C^2} \le \frac{n}{k_i} \ ,
\end{eqnarray*}
since $C \ge 5$. Suppose now that $\Exp \left[ d(v_j) \right] < 4 \log n$. Then, by Theorem~\ref{thm:chernoff} (setting $\mu = 8 \log n$ and $\delta = 8$),
\begin{eqnarray*}
 \Pr \left[ d(v_j) \ge 72 \log n \right] \le n^{-10} \ ,
\end{eqnarray*}
by the same calculation as above. Since $k_i \le \frac{n}{\log^2 n \cdot C}$ (Inequality~\ref{eqn:k-bound}), we have $d(v_j) \le \frac{n}{k_i}$, which
completes the proof.

\qed
\end{proof}

%\begin{lemma}
% The probability that Algorithm~1 aborts in Step~1 is at most $\frac{1}{n^2}$.
%\end{lemma}
%\begin{proof}
 %The probability that all values $r_i$ are different is:
 %\begin{eqnarray*}
 % \frac{n^3 \cdot (n^3-1) \cdot (n^3-2) \cdot \dots \cdot (n^3-n+1)}{n^{3n}} & > & \frac{(n^3-n)^n}{n^{3n}} = (1- \frac{1}{n^2})^n \le \frac{1}{e^n} \ .
 %\end{eqnarray*}
%Fix a vertex $v$ with its choice $r_v$. The probability that a vertex $u \neq v$ has chosen $r_u$ with $r_u = r_v$ is $\frac{1}{n^4}$. Therefore, by the union bound,
%the probability that any of the vertices $V \setminus \{v\}$ has chosen the same value $r_v$ is at most $\frac{1}{n^3}$. Thus, again by the union bound, the probability that there
%is at least one collision is at most $\frac{1}{n^2}$. \qed
%\end{proof}

\noindent \textbf{Theorem~\ref{thm:main} (restated)} {\em 
Algorithm~2 operates in $\Order(\log \log \Delta)$ rounds in the {\sf{CONGESTED-CLIQUE}} model and outputs a maximal independent set with high probability.}

\begin{proof}
 Concerning the runtime, Step~1 of the algorithm requires $\Order(1)$ communication rounds. Observe that every iteration of the while-loop requires $\Order(1)$ 
 rounds. The while-loop terminates in $\Order(\log \log \Delta)$ rounds with high probability, by Corollary~\ref{cor:runtime}. Since Ghaffari's algorithm 
 requires $\Order(\log \log \Delta') = \Order(\log \log \Delta)$ rounds, where $\Delta'$ is the maximum degree in the residual as computed in the last iteration 
 of the while-loop (or in case $\Delta < \log^4 n$ then $\Delta' = \Delta$), the overall runtime is bounded by $\Order(\log \log \Delta)$.
 
 Concerning the correctness of the algorithm, the only non-trivial step is the collection of graph $H_i$ at vertex $v_1$. This is achieved using Lenzen's 
 routing protocol, which can be used since we proved in Lemma~\ref{lem:edge-bound} that graph $H_i$ has at most $n$ vertices with high probability. \qed
\end{proof}

%  \begin{algorithmic}[1]
%  \STATE \textbf{Nodes $V$ agree on a uniform random order: }
%  \STATE $\quad$ Every vertex $v$ selects a random number $r_v \in \{1, \dots, n^3 \}$ and broadcasts it to all other nodes
%  \STATE $\quad$ Every vertex $v$ sorts $V$ with increasing $r_u$ values, if two nodes received the same number then abort
%  \STATE $\quad$ Let $v_1, v_2, \dots, v_n$ be this order
%  \STATE \textbf{Process $V$ in $\Order(\log \log n)$ phases: }
%  \STATE $\quad$ Every vertex $v$ initializes $u_v \gets $ \textbf{true} \COMMENT{initially, every node is uncovered}
%  \STATE $\quad$ \textbf{for} $j \gets 1, \dots, \lceil \log \log n \rceil$
%  \STATE $\quad$ All vertices $v_i$ with $u_{v_i} = $ \textbf{true} and $n_{j-1} = n^{1-\frac{1}{2^{j-1}}} / \log n < 
%  i \le n^{1-\frac{1}{2^j}} / \log n = n_j$ \textbf{do}:
%  \STATE $\quad\quad$ Use Lenzen's protocol to send all edges of $G[n_{j-1}+1, n_j]$  to $v_1$
%  \STATE $\quad\quad$ $v_1$ simulates the sequential Greedy algorithm, Algorithm~\ref{alg:1}, on vertices  $[n_{j-1}+1, n_j]$
%  \STATE $\quad\quad$ $v_1$ informs vertices that are chosen into $I$
%  \STATE $\quad\quad$ Chosen nodes inform their neighbors that they are chosen
%  \STATE $\quad\quad$ Every chosen node or neighbor of a chosen node sets $u_v =$ \textbf{false}
%  \end{algorithmic}

\section{Conclusion} \label{sec:conclusion}
In this paper, we gave a $\Order(\log \log \Delta)$ rounds MIS algorithm that runs in the \CONGESTED model. We simulated the sequential random 
order \textsc{Greedy} algorithm, exploiting the residual sparsity property of \textsc{Greedy}. 

It is conceivable that the round complexity can be reduced further - there are no lower bounds known for MIS in the \CONGESTED model.
Results on other problems, such as the minimum weight spanning tree problem where the $\Order(\log \log n)$ rounds algorithm 
of Lotker et al. \cite{l03} has subsequently been improved to $\Order(\log \log \log n)$ rounds \cite{hppss15}, $\Order(\log^* n)$ rounds \cite{gp16}, and finally to $\Order(1)$ rounds \cite{jn18},
give hope that similar improvements may be possible for \textsc{MIS} as well.
Can we simulate other centralized Greedy algorithms in few rounds in the \CONGESTED model? 

\subsection*{Acknowledgements} The author thanks Amit Chakrabarti, Anthony Wirth, and Graham Cormode for discussions about the residual sparsity property of the clustering 
algorithm given in \cite{acgmw15}.

\bibliography{k17}
\end{document}